\documentclass[11pt]{article}
\usepackage{amsmath}
\usepackage{amssymb,amsbsy,amsthm,dsfont}
\usepackage{graphicx}
\usepackage[dvipsnames]{xcolor}
\usepackage{caption}
\usepackage{subcaption}
\usepackage{enumerate}
\usepackage{cases}
\usepackage[margin=1in]{geometry}
\usepackage{hyperref}

\usepackage{algorithm}
\usepackage[noend]{algpseudocode}
\usepackage{bm}
\usepackage{bbm}

\allowdisplaybreaks[1]

\numberwithin{equation}{section}

\DeclareMathOperator{\R}{\mathbb{R}} 

\newcommand{\eps}{\varepsilon} 

\newcommand{\Beta}{\mathrm{Beta}} 

\newcommand{\ind}[1]{\mathbbm{1}\left\{#1\right\}}

\newcommand{\wt}{\widetilde} 

\newtheorem{theorem}{Theorem}[section]
\newtheorem{lemma}[theorem]{Lemma}

\newtheorem{example}[theorem]{Example}

\begin{document}

\title{Network disruption: maximizing disagreement and polarization in social networks}
\author{
	Mayee F. Chen 
	\thanks{Stanford University; \texttt{mfchen@stanford.edu}.} 
	\and
	Mikl\'os Z.\ R\'acz
	\thanks{Princeton University; \texttt{mracz@princeton.edu}. Research supported in part by NSF grant DMS 1811724 and by a Princeton SEAS Innovation Award.}
}
\date{\today}

\maketitle


\begin{abstract} 
Recent years have seen a marked increase in the spread of misinformation, a phenomenon which has been accelerated and amplified by social media such as Facebook and Twitter. While some actors spread misinformation to push a specific agenda, it has also been widely documented that others aim to simply disrupt the network by increasing disagreement and polarization across the network and thereby destabilizing society. Popular social networks are also vulnerable to large-scale attacks. Motivated by this reality, we introduce a simple model of \emph{network disruption} where an adversary can take over a limited number of user profiles in a social network with the aim of maximizing disagreement and/or polarization in the network. 

We investigate this model both theoretically and empirically. We show that the adversary will always change the opinion of a taken-over profile to an extreme in order to maximize disruption. We also prove that an adversary can increase disagreement / polarization at most linearly in the number of user profiles it takes over. Furthermore, we present a detailed empirical study of several natural algorithms for the adversary on both synthetic networks and real world (Reddit and Twitter) data sets. These show that even simple, unsophisticated heuristics, such as targeting centrists, can disrupt a network effectively, causing a large increase in disagreement / polarization. Studying the problem of network disruption through the lens of an adversary thus highlights the seriousness of the problem. 
\end{abstract}


\section{Introduction} \label{sec:intro} 
 
 Recent years have seen a marked increase in the spread of misinformation, 
a phenomenon which has been accelerated and amplified by social media such as Facebook and Twitter. 
This problem has been widely studied empirically~\cite{AG17,fourney2017geographic,guessfakenews,spangher2018analysis,vosoughi2018spread}. 
By and large, the main solution proposed to tackle the spread of misinformation is to develop automated fake news detection tools (e.g.,~\cite{tacchini2017some}). 
However, there are huge challenges to overcome to make this viable. To start, simply defining what is false vs.\ true is often controversial and by now has been hugely politicized. 
Moreover, rapid advances in machine learning have made possible the creation of fake audio and video that are convincingly realistic, 
hence the problem of detection will only become worse in the coming years.

In this paper we consider a completely different angle. 
While some actors spread misinformation to push a specific agenda, it has also been widely documented 
\cite{facebook,USHPSCI} that others aim to simply disrupt the network by increasing disagreement and polarization across the network, thereby destabilizing society. 
 Popular social networks are also vulnerable to large-scale attacks---in September 2018 it was revealed that nearly 50 million Facebook users were compromised in a data breach where attackers had the ability to take over accounts~\cite{guardian}. 
Motivated by this reality, we introduce a simple model of \emph{network disruption} where an adversary can take over some user profiles in a social network with the aim of maximizing disagreement and/or polarization in the network. 

Does the adversary have to be sophisticated to cause significant disruption? Or can they achieve their goal via simple, unsophisticated heuristics? 
How do the answers to these questions depend on properties of the underlying social network? 
The goal of this paper is to study such questions. 
By studying the problem through the lens of an adversary we are able to gain insight into the seriousness of the problem. 

\subsection{An adversarial model of network disruption} \label{sec:model} 

Our key conceptual contribution is the introduction of a novel adversarial model of network disruption. 
We model the underlying social network as a weighted graph $G = (V,E,w)$, 
where $V$ is the set of vertices, corresponding to the users of the social network, 
$E$ is the set of edges, connecting users who know each other, 
and $w : E \to [0,1]$ is a weight function on the edges that describes the strength of the ties between users. 
Now consider a topic that everyone has an opinion about---gun ownership, the amount of taxation, or your favorite controversial topic. 
We assume that everyone has an \emph{innate opinion} about this topic and that this opinion can be quantified by a number in the interval $[0,1]$; 
for instance, $0$ corresponds to strict gun control while $1$ corresponds to no gun control. 
The innate opinions are denoted by $s = \left\{ s_{v} \right\}_{v \in V} \in \left[ 0, 1 \right]^{V}$. 

People interact with their acquaintances on the social network and exchange opinions. 
As a result, their opinions
evolve and finally reach an \emph{equilibrium}, 
which we denote 
by $z = \left\{ z_{v} \right\}_{v \in V} \in \left[ 0, 1 \right]^{V}$. 
To be specific, in this paper we consider a simple model of opinion dynamics---known as the Friedkin-Johnsen model~\cite{friedkin1990social}---where users iteratively update their opinions by taking a weighted average of the opinions of their friends and their innate opinion. This results in the equilibrium opinions being 
$z = \left( I + L \right)^{-1} s$, where $I$ is the identity matrix and $L$ is the (weighted) Laplacian matrix. 
We emphasize that, while we focus on the Friedkin-Johnsen model in this paper, 
everything we consider can be studied for other opinion dynamics models as well.

The equilibrium opinions $z$ have various properties that we care about. 
Following~\cite{musco2018minimizing}, we introduce the following two important quantities. 
\emph{Disagreement} is defined as 
\begin{equation}\label{eq:D}
D \equiv D(z) := \sum_{\left( u,v \right) \in E} w_{u,v} \left( z_{u} - z_{v} \right)^{2};
\end{equation}
this measures how much acquaintances disagree in their opinions, globally across the network. 
\emph{Polarization} is defined as 
\begin{equation}\label{eq:P}
P \equiv P(z) := \sum_{v \in V} \left( z_{v} - \overline{z} \right)^{2}, 
\end{equation}
where $\overline{z} := \tfrac{1}{\left| V \right|} \sum_{v \in V} z_{v}$ is the mean opinion; 
in other words, $P$ is the variance of the opinions, multiplied by the number of vertices. 

\begin{figure}[ht!]
\includegraphics[width=\textwidth]{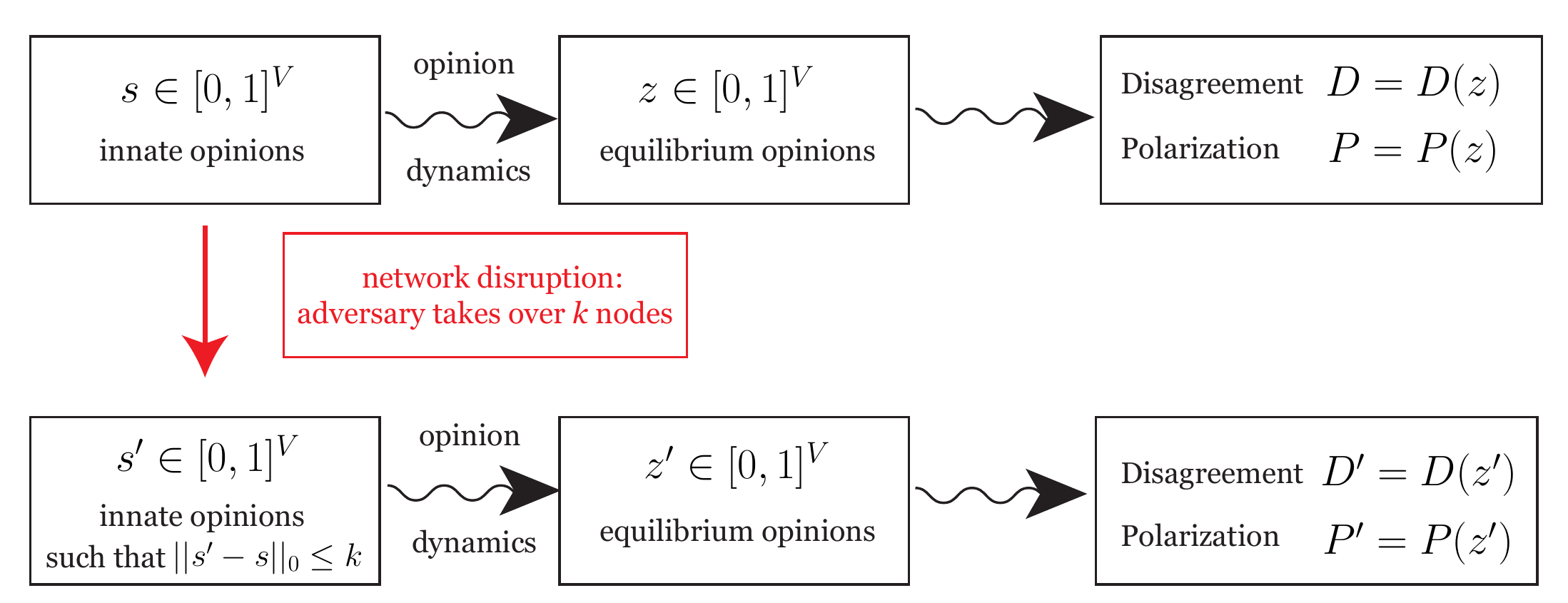}
\caption{\textbf{Schematic of the adversarial model of network disruption.} \emph{Top:} On a particular topic everyone has an innate opinion, resulting in the innate opinion vector $s \in [0,1]^{V}$. These are mapped to equilibrium opinions $z \in [0,1]^{V}$ via the opinion dynamics. The equilibrium opinions give rise to natural quantities: disagreement $D$ and polarization $P$. \emph{Bottom:} The adversary can take over at most $k$ nodes in the network and change their innate opinions, resulting in the new innate opinion vector $s' \in [0,1]^{V}$. The opinion dynamics are unchanged, resulting in new equilibrium opinions $z' \in [0,1]^{V}$, and subsequently new values of disagreement $D'$ and polarization~$P'$. The goal of the adversary is to \emph{maximize} disagreement and/or polarization.}
\label{fig:schematic}
\end{figure}

\textbf{Modeling the adversary.} We now turn to modeling network disruption, which is the key new idea introduced in the paper. 
We consider an adversary who can take over $k$ nodes of the network and modify the innate opinions of these nodes arbitrarily. 
That is, the adversary can select $s' \in [0,1]^{V}$ such that $\left\| s' - s \right\|_{0} \leq k$. 
For instance, this can model a hacker who takes over a set of Facebook profiles. 
The adversary does not want to raise suspicion and so the opinion dynamics remain unchanged. 
Therefore, assuming the Friedkin-Johnsen model, the resulting equilibrium opinions will be $z' = (I + L)^{-1} s'$ 
and these will result in new values of disagreement $D'$ and polarization $P'$. 
The goal of the adversary is to pick $s'$ in such a way that \emph{maximizes} disagreement $D'$ or polarization~$P'$. 
See Figure~\ref{fig:schematic} for an illustration.

\textbf{Questions and challenges.} 
What is the optimal solution for the adversary? 
That is, how should they pick the set of $k$ vertices to hack, and how should they set the innate opinions of hacked vertices? 

Our first result shows that any optimal solution will set the innate opinions to an extreme; 
that is, if $s'_v \neq s_v$ then $s'_v \in \left\{0,1\right\}$. 
Thus a brute force approach can find an optimal solution by checking all 
$\binom{n}{k} 2^{k}$ possibilities, where $n$ denotes the number of vertices. 
This is not feasible when $k$ is large--- so is there an efficient (polynomial time in~$n$) algorithm to find an optimal solution? 
One can show that the function that we are maximizing is not submodular (see Section~\ref{sec:convexity}) and hence off-the-shelf methods/results do not apply directly.

Regardless if they can efficiently find an optimum or not, 
it may be argued that in some cases knowing all the innate opinions exactly is unrealistic, and in other cases knowing the entire social network structure is difficult. 
Can the adversary cause significant disruption knowing only the network structure and nothing (or close to nothing) about the innate opinions, and vice versa? 
Can simple heuristics perform well? 
How do the answers to these questions depend on properties of the underlying social network? 
We investigate such questions in this paper. 

\subsection{Our results} \label{sec:results} 

Our first result is intuitive: no matter which set of vertices the adversary chooses, the optimal way to modify the innate opinions of these nodes is to set them to one of the two extremes: $0$ or $1$. In particular, we have the following result. 

\begin{theorem}\label{thm:convexity}
Consider the problem setup as above, with the adversary maximizing either disagreement, polarization, or a conical combination of these two (i.e., a linear combination with nonnegative coefficients). 
Assume that $G$ has no isolated vertices. 
Let $s'$ be an optimum vector of innate opinions, given the constraints. 
For every $v \in V$, if $s_{v}' \neq s_{v}$, then $s_{v}' \in \left\{ 0, 1 \right\}$.
\end{theorem}

This result follows from the convexity of the objective functions, together with the fact that the adversary is maximizing the objective function. 
This implies that if the adversary has a budget of~$k$ (i.e., it can take over at most $k$ nodes), then a brute force approach can find an optimal solution by checking all $\binom{n}{k} 2^{k}$ possibilities, where $n$ denotes the number of nodes. For constant $k$ this gives a polynomial-time algorithm, but it performs poorly as $k$ grows. 
In fact, we conjecture that solving the optimization problem of the adversary is computationally hard when $k$ is large (e.g., $k = n^{\eps}$). 

Next, we examine quantitatively the effect that the adversary can have on disagreement and polarization, both theoretically and empirically. 
First, we prove that the adversary can only increase disruption linearly in $k$. 
Specifically, for the polarization objective we show that the increase is always bounded above by $3k$; this is the content of the following theorem.

\begin{theorem}[Upper bound on the increase in polarization]\label{thm:P_bound}
Let $G$ be a weighted graph and $s$ a vector of innate opinions such that the resulting equilibrium opinion vector $z$ has polarization~$P$. 
Suppose that the adversary has a budget of $k$; 
that is, the adversary may select $s' \in \left[ 0, 1 \right]^{V}$ such that $\left\| s' - s \right\|_{0} \leq k$. 
Let $P'$ be the polarization of the resulting equilibrium opinion vector $z' = \left( I + L \right)^{-1} s'$. 
Then 
\[
P' \leq P + 3k.
\]
\end{theorem}

For the disagreement objective, our result gives a bound of $8 d_{\max} k$, where $d_{\max}$ is the (weighted) maximum degree. Thus for bounded-degree graphs this is still $O(k)$. 

\begin{theorem}[Upper bound on the increase in disagreement]\label{thm:D_bound}
Let $G$ be a weighted graph and $s$ a vector of innate opinions such that the resulting equilibrium opinion vector $z$ has disagreement~$D$. 
Suppose that the adversary has a budget of $k$; 
that is, the adversary may select $s' \in \left[ 0, 1 \right]^{V}$ such that $\left\| s' - s \right\|_{0} \leq k$. 
Let $D'$ be the disagreement of the resulting equilibrium opinion vector $z' = \left( I + L \right)^{-1} s'$. 
Then 
\[
D' \leq D + 8 d_{\max} k,
\]
where $d_{\max} := \max_{v \in V} \sum_{u \in V} w_{v,u}$ is the (weighted) maximum degree. 
\end{theorem}

These results lead to a natural question: can the adversary achieve an increase in these objective functions that grows linearly with $k$? 
We show empirically, on both synthetic and real data sets, that this is indeed the case. 

We first consider a greedy algorithm, where the adversary iteratively selects nodes to take over, in each iteration choosing the node, together with one of the two extreme opinions, that maximizes the objective function. 
While this greedy algorithm is natural, it also uses detailed information: specifically, it assumes knowledge of the network $G$ and the innate opinions $s$. 
Since this may be unrealistic in practice, we also consider simpler heuristics for the adversary. 

One such heuristic, which we term the ``mean opinion'' heuristic, is to choose the node whose (innate) opinion is closest to the mean and change it to one of the two extremes (either by optimizing this choice or just randomly). 
Such a heuristic can easily be implemented approximately by an adversary, since often it is possible to deduce whether someone has a centrist opinion by using extra information available about the node. 
Another heuristic that we consider focuses on a simple function of the underlying graph structure: 
iteratively choosing the largest degree nodes (in either a weighted or unweighted sense) and changing their opinion to one of the two extremes. 
We also compare all the algorithms to a random baseline, where the adversary selects nodes randomly and changes their opinions to random extremes.

We evaluate these algorithms on both synthetic and real data sets. For synthetic networks we use three common probabilistic generative models: Erd\H{o}s-R\'enyi random graphs, the preferential attachment model, and the stochastic block model. 
We also study Reddit and Twitter data sets that were collected in~\cite{de2014} and subsequently studied in~\cite{musco2018minimizing}.

Our main empirical finding is that in almost all settings---meaning, a network (synthetic or real, as above), an algorithm (from the ones described above), and an objective function (disagreement, polarization, or a conical combination)---the adversary succeeds in increasing its objective function linearly in $k$.  
The rate of increase depends on the details: the greedy algorithm performs best among these options, but the mean opinion heuristic is often not far behind. Even the random baseline gives a linear increase in $k$ in several (though not all) settings. The details of all our experiments are in Section~\ref{sec:algorithms}. Our code and data are publicly available at 
\url{https://github.com/mayee107/network-disruption}.

The key take-away from our results is that an adversary can significantly disrupt a network using simple, unsophisticated methods. 
This mirrors recent findings analyzing real-world data; for instance, the authors in~\cite{boyd2018characterizing} conclude that the Internet Research Agency's operations to interfere with the 2016 U.S.~presidential election ``were largely unsophisticated''. 
This adversarial approach thus highlights the seriousness of the problem and motivates further research into addressing it.

\subsection{Related work} \label{sec:related} 

Opinion dynamics have been used in various disciplines to model social learning (see, e.g., the survey~\cite{mossel2017opinion}). 
In seminal work, the DeGroot model describes how individuals reach a consensus through stochastic interactions~\cite{degroot1974reaching}.  
Friedkin and Johnsen extended this model to incorporate individuals' intrinsic beliefs and prejudices~\cite{friedkin1990social}. 
In the Friedkin-Johnsen model, all agents have individual innate opinion values, 
and as time goes on, agents interact with each other, updating their opinions to be a weighted average of their innate opinion and the neighboring agents' opinions. Eventually, opinions converge to an equilibrium, which is a non-constant function of the innate opinions. 
This latter property is an important reason why we use the Friedkin-Johnsen model for opinion dynamics in this paper, in addition to its simplicity. 
The Friedkin-Johnsen model can be extended in a variety of ways, for instance to incorporate stubbornness and susceptibility to persuasion~\cite{abebe2018opinion}.

Several recent works have studied various network interventions to influence opinions in certain ways. 
Gionis, Terzi, and Tsaparas~\cite{gionis2013opinion} studied opinion maximization in social networks, which corresponds to pushing a specific agenda. 
Abebe~et~al.~\cite{abebe2018opinion} study a similar problem (opinion maximization or minimization), but where interventions happen at the level of susceptibility to persuasion. 

In contrast, the work of Musco, Musco, and Tsourakakis~\cite{musco2018minimizing}---which serves as the starting point of our work---studies polarization and disagreement, which are quite different objectives. The goal of their work is to \emph{minimize} these quantities by changing the underlying network topology.

Our key conceptual contribution is to study the \emph{opposite} objective: \emph{maximizing} polarization and disagreement. This corresponds to an adversarial perspective, which is motivated by recent developments over the past few years: malicious actors have increasingly been working towards disrupting networks by increasing disagreement and polarization, thereby destabilizing society~\cite{facebook, guardian, USHPSCI, boyd2018characterizing}. Also, the specific intervention we consider is taking over nodes of a network and modifying their (innate) opinions.

\emph{Note added.} While finishing this paper, the paper~\cite{gaitonde2020} was posted to arXiv, with similar ideas. However, their focus is on the special case when society initially has a consensus (i.e., $s=0$), and this is perturbed by an adversary that can modify the entire innate opinion vector. They formalize the constraint on the adversary as an $L_{2}$-norm bound, whereas we use the constraint $\left\| s' - s \right\|_{0} \leq k$, which has a clear interpretation in the adversary taking over at most $k$ nodes of the network. 

\subsection{Outline} \label{sec:outline} 

The rest of the paper is organized as follows. In Section~\ref{sec:problem} we detail the problem setup for clarity. In Section~\ref{sec:convexity} we prove Theorem~\ref{thm:convexity}, and then in Section~\ref{sec:bounds} we prove Theorems~\ref{thm:P_bound} and~\ref{thm:D_bound}. 
In Section~\ref{sec:algorithms} we empirically analyze several efficient algorithms and heuristics for the adversary, both on synthetic networks generated according to common probabilistic models and on Twitter and Reddit data sets. 
We conclude in Section~\ref{sec:discussion} with discussion and questions for future work. 

\section{Problem setup} \label{sec:problem} 

In this section we detail the problem setup for clarity. 
We fix a weighted graph $G=(V,E,w)$ which represents the social network. 
Let $n=|V|$ denote the number of vertices (we often write $[n]$ for the vertex set) and let $m = |E|$ denote the number of edges. 
For convenience we define the weight function on all pairs of nodes, with $0 < w_{i,j} \le 1$ if $(i,j) \in E$ and $w_{i,j} = 0$ otherwise. 
We also set $w_{i,i} = 0$ for all $i \in V$. 

Let $d_{i} = \sum_{j \in V} w_{i,j}$ denote the (weighted) degree of node $i$ and let $D$ be the diagonal matrix with entries $d_{1}, \ldots, d_{n}$ on the diagonal. Let $A$ denote the (weighted) adjacency matrix of $G$, with $A_{i,j} := w_{i,j}$ for $i, j \in V$. 
Let $L = D - A$ denote the weighted combinatorial Laplacian of $G$, which  we refer to just as the Laplacian of $G$. 
Finally, let $\vec{1}$ denote the all-ones vector. 

\textbf{Opinion dynamics.} 
Let $s = \left( s_{1}, \ldots, s_{n} \right) \in \left[ 0, 1 \right]^{n}$ denote the vector of innate opinions. 
In the Friedkin-Johnsen model of opinion dynamics~\cite{friedkin1990social}, agents interact with each other as time goes on, updating their opinions to be a weighted average of their innate opinion and the neighboring agents' opinions. 
Formally, if $z_{i}^{(t)}$ denotes the opinion of node $i$ at time $t$ (where $t \in \left\{ 0, 1, 2, \ldots \right\}$), then initially $z_{i}^{(0)} = s_{i}$ and the update for $t \geq 0$ is given by 
\[
z_{i}^{(t+1)} 
= 
\frac{s_{i}+\sum_{j \in V} w_{i,j} z_{j}^{(t)}}{1+\sum_{j \in V} w_{i,j}}.
\]
As $t \to \infty$, the vector of opinions converges to an equilibrium vector $z$ that satisfies 
\begin{equation}\label{eq:z}
z = \left( I + L \right)^{-1} s,
\end{equation}
where $I$ is the $n \times n$ identity matrix. 

\textbf{Disagreement and polarization.} 
Following~\cite{musco2018minimizing}, we study 
the disagreement $D(z)$ and the polarization $P(z)$ of a vector of opinions $z$; see~\eqref{eq:D} and~\eqref{eq:P} for the definitions. 
Note that since the equilibrium opinion vector $z$ is a function of the innate opinion vector $s$, disagreement $D$ and polarization $P$ can be considered functions of $s$ as well, in which case we will denote them by $D(s)$ and $P(s)$, respectively. 
When clear from the context, we may denote these by just $D$ and $P$. 
We also study linear combinations of these two quantities.

\textbf{The objectives of the adversary.} 
We are now ready to mathematically formulate our original questions as three optimization problems with varying objective functions. 
For any weighted graph $G$, innate opinions $s$, and budget $k \in \mathbb{N}$, the adversary aims to determine the optimal modified innate opinion vector $s'$ according to the following. 

\begin{itemize}
\item \textbf{Problem 1: Disagreement}
\begin{equation}
\begin{aligned}
& \text{maximize}
& & D\left(z'\right) \\
& \text{subject to}
& & z' = (I + L)^{-1} s', \\
& & & s' \in [0, 1]^n, \\
& & & ||s' - s||_0 \le k.
\end{aligned} \label{opt:D}
\end{equation}

\item \textbf{Problem 2: Polarization}
\begin{equation}
\begin{aligned}
& \text{maximize}
& & P\left(z'\right) \\
& \text{subject to}
& & z' = (I + L)^{-1} s', \\
& & & s' \in [0, 1]^n, \\
& & & ||s' - s||_0 \le k.
\end{aligned} \label{opt:P}
\end{equation}

\item \textbf{Problem 3: Weighted Sum}
\begin{equation}
\begin{aligned}
& \text{maximize}
& & P\left(z'\right) + \lambda \frac{n}{m} D \left( z' \right) \\
& \text{subject to}
& & z' = (I + L)^{-1} s', \\
& & & s' \in [0, 1]^n, \\
& & & ||s' - s||_0 \le k.
\end{aligned} \label{opt:Sum}
\end{equation}
\end{itemize}

Note that in \eqref{opt:Sum}, we introduce $\lambda$ as a parameter to describe the relative importance of disagreement versus polarization to the adversary. For this weighted sum index, we have scaled disagreement by $\frac{|V|}{|E|} = \frac{n}{m}$ so that polarization and disagreement are considered with equal weight and have the same order of growth in our analysis when $\lambda = 1$.

\section{Convexity and choosing extreme opinions} \label{sec:convexity} 

For all three optimization problems, the set of constraints do not form a convex set due to the constraint $||s' - s||_0 \le k$. However, we prove that all of the objective functions are convex in $s'$, which implies that $s_i' \in \{0, 1\}$ for all vertices $i$ where $s_i' \neq s_i$.

\begin{lemma}\label{lem:D_convexity}
Disagreement is convex in $s'$. That is, the function $s' \mapsto D \left( s' \right)$ is convex. 
\end{lemma}

\begin{proof}
Disagreement can be written in quadratic form as $z'^T L z'$. Noting that $I + L $ is symmetric and using~\eqref{eq:z}, $D$ can be expressed as
\[
D\left( s' \right) = z'^T L z' = \left(\left(I + L\right)^{-1} s'\right)^T L \left(\left(I + L\right)^{-1} s'\right) = s'^T \left(I + L \right)^{-1} L \left(I + L \right)^{-1} s'.
\]
The Laplacian matrix $L$ is positive semidefinite and symmetric, so $L$ can be written as $L = B^T B$ for some matrix $B \in \mathbb{R}^{n \times n}$. Therefore, $(I + L)^{-1} L (I + L)^{-1} = (I + L)^{-1} B^T B (I + L)^{-1} = (B(I + L)^{-1})^T (B(I + L)^{-1})$, so $(I + L)^{-1} L (I + L)^{-1}$ is also positive semidefinite. Thus we can write $D$ as a quadratic form in terms of $s'$, with a positive semidefinite matrix, so $D\left(s'\right)$ is convex in $s'$. 
\end{proof}

\begin{lemma}\label{lem:P_convexity}
Polarization is convex in $s'$. That is, the function $s' \mapsto P \left( s' \right)$ is convex. 
\end{lemma}

\begin{proof}
For notational convenience we drop all apostrophes from the notation. 
For a vector $x \in \R^{n}$ let $\wt{x} := x - \overline{x} \vec{1}$ denote the centered vector.  
With this notation we have $P(z) = \wt{z}^{T} \wt{z}$.

Observe that $L \vec{1} = 0$, and so $(I+L)\vec{1} = \vec{1}$ and $(I+L)^{-1} \vec{1} = \vec{1}$. 
Using~\eqref{eq:z} this implies that 
$\overline{z} = \tfrac{1}{n} z^{T} \vec{1} 
= \tfrac{1}{n} z^{T} (I+L)^{-1} \vec{1} 
= \tfrac{1}{n} s^{T} \vec{1} = \overline{s}$. 
In words, the mean equilibrium opinion is the same as the mean innate opinion. 
This, in turn, implies that $\wt{z} = (I+L)^{-1} \wt{s}$. 
With this notation we have that 
\[
P(z) = \wt{z}^{T} \wt{z} 
= \wt{s}^{T} \left( \left( I + L \right)^{-1} \right)^{2} \wt{s}.
\]
For a vector $x \in \R^{n}$ define $f(x) := x^{T} \left( \left( I + L \right)^{-1} \right)^{2} x$ and $g(x) := x - \overline{x} \vec{1} = \wt{x}$. 
Note that $\left( \left( I + L \right)^{-1} \right)^{2}$ is positive semidefinite, since it is the square of $\left( I + L \right)^{-1}$, which is positive semidefinite and symmetric. 
This implies that $f$ is convex, since it is a quadratic form with a positive semidefinite matrix. 
Note also that for any two vectors $x, y \in \R^{n}$ and $\alpha \in [0,1]$ we have that 
$g(\alpha x + (1-\alpha) y) = \alpha g(x) + (1-\alpha) g(y)$. 
Therefore the convexity of $P = f \circ g$ follows directly from the convexity of $f$.
\end{proof}

An immediate consequence of Lemmas~\ref{lem:D_convexity} and~\ref{lem:P_convexity} is that any conical combination of disagreement and polarization is convex in $s'$. 
This is because convexity is preserved by scaling with a positive constant, as well as across addition. 

\begin{proof}[Proof of Theorem~\ref{thm:convexity}]
Lemmas~\ref{lem:D_convexity} and~\ref{lem:P_convexity} show that the adversary's optimization problem is a convex maximization problem in $s'$. 
Moreover, if $G$ has no isolated vertices then this is a strictly convex maximization problem. 
Therefore any coordinate of $s$ that is changed in $s'$ must be changed to an extreme: $0$ or $1$. 
\end{proof}

We conclude this section by a simple example that shows that the objective functions we are considering are not submodular. 

\begin{example}[A single edge]
Consider a graph with two nodes, denoted $1$ and $2$, with an edge between them with weight $w_{1,2} = 1$. Suppose that the innate opinions are initially centrist: $s_{1} = s_{2} = 1/2$. In this case the equilibrium opinions are also centrist: $z_{1} = z_{2} = 1/2$, leading to no disagreement or polarization: $D(z) = P(z) = 0$. 

If an adversary has a budget of $k = 1$, they will change the innate opinion of a(n arbitrary) node to an (arbitrary) extreme: $s_{1}' = 0$, $s_{2}' = 1/2$. This results in the equilibrium opinions 
$z_{1}' = 1/6$ and $z_{2}' = 1/3$, 
giving disagreement 
$D\left( z' \right) = 1/36$ 
and polarization 
$P\left( z' \right) = 1/72$.

If an adversary has a budget of $k=2$, 
they will change the innate opinions to opposite extremes: 
$s_{1}'' = 0$, $s_{2}'' = 1$. 
This results in the equilibrium opinions 
$z_{1}'' = 1/3$ and $z_{2}'' = 2/3$, 
giving disagreement 
$D \left( z'' \right) = 1/9$ 
and polarization 
$P \left( z'' \right) = 1/18$.

For both disagreement and polarization the increase in the second step is greater than the increase in the first step, and hence these objective functions are not submodular. 
\end{example}

\section{Bounds on network disruption} \label{sec:bounds} 

In this section we prove Theorems~\ref{thm:P_bound} and~\ref{thm:D_bound}.  
We start with a preliminary lemma which gives a bound on the $L_{1}$-norm of the difference between the modified equilibrium opinion vector $z'$ and the original equilibrium opinion vector $z$. 

\begin{lemma}\label{lem:z_diff_L1}
Let $s$ be the original innate opinion vector and 
let $s'$ be the modified innate opinion vector, 
satisfying $\left\| s' - s \right\|_{0} \leq k$. 
Let $z$ and $z'$ be the respective equilibrium opinion vectors. 
Then 
\[
\left\| z' - z \right\|_{1} \leq k.
\]
\end{lemma}

\begin{proof}
By~\eqref{eq:z} we have that 
\[
\left\| z' - z \right\|_{1} 
= \left\| \left( I + L \right)^{-1} \left( s' - s \right) \right\|_{1} 
\leq \sum_{i=1}^{n} \sum_{a=1}^{n} \left| \left( s'_{a} - s_{a} \right) \left( I + L \right)^{-1}_{ia} \right| 
= \sum_{i=1}^{n} \sum_{a=1}^{n} \left| s'_{a} - s_{a} \right| \left( I + L \right)^{-1}_{ia}, 
\]
where the inequality is due to the triangle inequality and the final equality is because the entries of $\left( I+L \right)^{-1}$ are nonnegative. 
Without loss of generality, assume that nodes $1, \ldots, k$ comprise the set of nodes taken over by the adversary. 
Since $s_{i} \in \left[ 0, 1 \right]$, we must have $\left| s_{i}' - s_{i} \right| \leq 1$ for all $i$. Thus 
\[
\left\| z' - z \right\|_{1}  
\leq 
\sum_{i=1}^{n} \sum_{a=1}^{k} \left| s'_{a} - s_{a} \right| \left( I + L \right)^{-1}_{ia} 
\leq 
\sum_{i=1}^{n} \sum_{a=1}^{k} \left( I + L \right)^{-1}_{ia}.
\]
Now interchanging the order of summation we have that 
\[
\sum_{i=1}^{n} \sum_{a=1}^{k} \left( I + L \right)^{-1}_{ia} 
= 
\sum_{a=1}^{k} \sum_{i=1}^{n} \left( I + L \right)^{-1}_{ia} 
= \sum_{a=1}^{k} 1 = k.
\]
Here we used the fact that the column sums of $\left( I + L \right)^{-1}$ are all equal to $1$, which follows from the fact that $\left( I + L \right)^{-1} \vec{1} = \vec{1}$ (shown in Section~\ref{sec:convexity}) and that $\left( I + L \right)^{-1}$ is symmetric. 
\end{proof}

\subsection{Bound on the increase in polarization}

\begin{proof}[Proof of Theorem~\ref{thm:P_bound}] 
We first rewrite $P'$ in a way to make $P$ appear. This can be done by adding and subtracting under the square, and then expanding the square: 
\begin{align}
P' &= \sum_{i=1}^{n} \left( z_{i}' - \overline{z}' \right)^{2} 
= \sum_{i=1}^{n} \left( z_{i}' - z_{i} + z_{i} - \overline{z} + \overline{z} - \overline{z}' \right)^{2} \notag \\
&= P + \sum_{i=1}^{n} \left( z_{i}' - z_{i} \right)^{2} + n\left( \overline{z} - \overline{z}' \right)^{2} 
+ 2 \sum_{i=1}^{n} \left( z_{i}' - z_{i} \right) \left( z_{i} - \overline{z} \right) \notag \\ 
&\quad + 2 \sum_{i=1}^{n} \left( z_{i}' - z_{i} \right) \left( \overline{z} - \overline{z}' \right) 
+ 2 \sum_{i=1}^{n} \left( z_{i} - \overline{z} \right) \left( \overline{z} - \overline{z}' \right). \label{eq:P'_expanded2}
\end{align}
Since $\sum_{i=1}^{n} \left( z_{i} - \overline{z} \right) = 0$, 
the last term in~\eqref{eq:P'_expanded2} is zero. 
Since $\sum_{i=1}^{n} \left( z_{i}' - z_{i} \right) = n \left( \overline{z}' - \overline{z} \right)$, 
the first term in~\eqref{eq:P'_expanded2} is equal to 
$-2n \left( \overline{z} - \overline{z}' \right)^{2}$. 
Plugging this back into the display above we obtain that 
\begin{equation}\label{eq:P'_simplified}
P' = P + \sum_{i=1}^{n} \left( z_{i}' - z_{i} \right)^{2} 
+ 2 \sum_{i=1}^{n} \left( z_{i}' - z_{i} \right) \left( z_{i} - \overline{z} \right)
- 2n\left( \overline{z} - \overline{z}' \right)^{2}. 
\end{equation}
The last term in~\eqref{eq:P'_simplified} is nonpositive, so we may drop it. 
For the first sum in~\eqref{eq:P'_simplified}, note that $z_{i} \in \left[ 0, 1 \right]$ for every $i \in \left[ n \right]$, 
so $\left( z_{i}' - z_{i} \right)^{2} \leq \left| z_{i}' - z_{i} \right|$. 
Together with Lemma~\ref{lem:z_diff_L1} this shows that 
$\sum_{i=1}^{n} \left( z_{i}' - z_{i} \right)^{2} \leq k$. 
Finally, for the other sum in~\eqref{eq:P'_simplified}, using the bound $\left| z_{i} - \overline{z} \right| \leq 1$ we have that 
$\sum_{i=1}^{n} \left( z_{i}' - z_{i} \right) \left( z_{i} - \overline{z} \right) 
\leq \sum_{i=1}^{n} \left| z_{i}' - z_{i} \right| 
\leq k$. 
Altogether this shows that $P' \leq P + 3k$ as desired. 
\end{proof}

\subsection{Bound on the increase in disagreement}

\begin{proof}[Proof of Theorem~\ref{thm:D_bound}] 
We start by rewriting $D'$ in a way to make $D$ appear. This can be done by adding and subtracting under the square, and then expanding the square. In the following all summations over $i$ and $j$ go from $1$ to $n$, so we do not write this out further. 
\begin{align*}
D' &= \sum_{i,j} w_{i,j} \left( z_{i}' - z_{j}' \right)^{2} 
= \sum_{i,j} w_{i,j} \left( z_{i}' - z_{i} + z_{i} - z_{j} + z_{j} - z_{j}' \right)^{2} \\
&= D + \sum_{i,j} w_{i,j} \left\{ \left( z_{i}' - z_{i} \right)^{2} + \left( z_{j}' - z_{j} \right)^{2} + 2 \left( z_{i}' - z_{i} \right) \left( z_{j} - z_{j}' \right) + 2 \left( z_{i} - z_{j} \right) \left( z_{i}' - z_{i} + z_{j} - z_{j}' \right) \right\}.
\end{align*}
We now bound the four sums above. The first two sums are equal by symmetry, and we have that 
\begin{align*}
\sum_{i,j} w_{i,j} \left\{ \left( z_{i}' - z_{i} \right)^{2} + \left( z_{j}' - z_{j} \right)^{2} \right\} 
&= 2 \sum_{i,j} w_{i,j} \left( z_{i}' - z_{i} \right)^{2} 
= 2 \sum_{i} d_{i} \left( z_{i}' - z_{i} \right)^{2} \\
&\leq 2 d_{\max} \sum_{i} \left| z_{i}' - z_{i} \right| 
\leq 2 d_{\max} k,
\end{align*}
where we used Lemma~\ref{lem:z_diff_L1} for the last inequality 
and the fact that $\left| z_{i}' - z_{i} \right| \in [0,1]$ in the inequality before that. 
Next, using the inequality 
$\left( z_{i}' - z_{i} \right) \left( z_{j} - z_{j}' \right) \leq \left| z_{i}' - z_{i} \right|$ 
we have that 
\[
2 \sum_{i,j} w_{i,j} \left( z_{i}' - z_{i} \right) \left( z_{j} - z_{j}' \right) 
\leq 2 \sum_{i,j} w_{i,j} \left| z_{i}' - z_{i} \right| 
\leq 2 d_{\max} k.
\]
Finally, 
we use the bound 
$\left( z_{i} - z_{j} \right) \left( z_{i}' - z_{i} + z_{j} - z_{j}' \right) 
\leq \left| z_{i}' - z_{i} \right| + \left| z_{j}' - z_{j} \right|$ to obtain that 
\begin{align*}
2 \sum_{i,j} w_{i,j} \left( z_{i} - z_{j} \right) \left( z_{i}' - z_{i} + z_{j} - z_{j}' \right) 
&\leq 
2 \sum_{i,j} w_{i,j} \left( \left| z_{i}' - z_{i} \right| +  \left| z_{j} - z_{j}' \right| \right) \\
&= 
4 \sum_{i,j} w_{i,j} \left| z_{i}' - z_{i} \right| 
\leq 4 d_{\max} k.
\end{align*}
Putting everything together we obtain that $D' \leq D + 8 d_{\max} k$ as desired. 
\end{proof}

\section{Algorithms for the adversary} \label{sec:algorithms} 

In this section, we first discuss several efficient algorithms and heuristics for the adversary to select vertices and set their opinions. We then present results of their performance on synthetic networks generated according to three common probabilistic models: Erd\H{o}s-R\'enyi random graphs, the preferential attachment model, and the stochastic block model. Finally, we evaluate the heuristics on the Twitter and Reddit data sets that were collected in~\cite{de2014} and subsequently studied in~\cite{musco2018minimizing}. 

\subsection{Algorithms and heuristics}\label{subsec:heuristics}

We present six adversarial heuristics that are designed under varying levels of information available about the network structure and opinions. 
We start with a natural greedy algorithm and then turn to other simpler heuristics.

\textbf{Greedy Algorithm.} In light of our analysis on convexity, we propose an efficient greedy algorithm for selecting $s'$. In this algorithm, we maintain a set $\Omega$ of vertices that have already been selected by the adversary, whose modified opinion values are reflected in our update of $s'$. In each iteration $i$, we select a vertex and set its opinion to $0$ or $1$ to result in the greatest increase in the objective function, given that the last $i - 1$ opinions have already been picked and modified according to this algorithm.  We then add this vertex to $\Omega$, update $s'$, and repeat this method a total of $k$ times. If no modification results in an increase in the objective function at the $i$th iteration, with $i < k$, then we stop the procedure. 
Algorithm \ref{alg:greedy} describes in pseudocode our greedy heuristic for determining $s'$, which can be used on any of our three objective functions using the notation $f \in \{P, D, P + \lambda \frac{n}{m} D \}$. 
\emph{Note:} with a slight abuse of notation, 
we write simply $f(s)$ to denote the objective function $f$ applied to the equilibrium opinions $z$ that are obtained from the innate opinions $s$.

\begin{algorithm}
\caption{Greedy heuristic for budgeted disruption maximization}
\begin{algorithmic}
    \Procedure{Greedy}{$f \in \{P, D, P + \lambda \frac{n}{m}D \}$, $s$}
\label{alg:greedy}
    	\State $s' = s$
    	\State $F' = f(s')$
    	\State $\Omega = \emptyset$   
    	\For{$i = 1 $ to $k$}
    		\State $s^{(i)} = s'$
    	    \State $j^* = \underset{j \notin \Omega}{\text{argmax}} \; \underset{a \in \{0, 1\}}{\max} \{f(s^{(i)}): s_j^{(i)} = a\}$
    		\State $a^* = \underset{a \in \{0, 1\}}{\text{argmax}} \{f(s^{(i)}): s_{j^*}^{(i)} = a \}$
    		\State $s_{j^*}^{(i)} = a^*$
    		\If{$ f(s^{(i)}) \ge F'$}
                \State $s_{j^*}' = a^*$
                \State $F' = f(s')$
                \State $\Omega = \Omega \cup \{j^*\}$
            \Else
                \State \textbf{break}
    		\EndIf
        \EndFor
        \Return $s'$
    \EndProcedure
\end{algorithmic}
\end{algorithm}

To compute any value of $f$, we need the matrix $(I + L)^{-1}$, which is calculated once at the beginning of the algorithm. This takes $\mathcal{O}(n^{2.3727})$ time using optimized Coppersmith-Winograd-like algorithms~\cite{williams2012}. Excluding this one-time inevitable computation, the runtime of Algorithm~\ref{alg:greedy} is $\mathcal{O}(nk)$. In comparison, a brute force algorithm would have runtime of $\mathcal{O}\left(\binom{n}{k} 2^k \right)$. 

\paragraph{Other Heuristics.} While Algorithm~\ref{alg:greedy} may seem to be the method most in accordance with our problem setup, it requires us to calculate $f(s^{(i)})$ with each element of $s^{(i)}$ changed to $0$ and then~$1$. This algorithm requires full knowledge of the Laplacian $L$---and thus the full graph $G$---and the network's innate opinions $s$. This may not be practically feasible. We hence explore several other heuristics, applicable for any of the three objective functions, that are computationally simpler and may be possible for the adversary to use even with limited information. We set up each heuristic with the same notions of $\Omega$ and $s'$, and each heuristic has two separate parts---picking the vertex and then setting the opinion---that are invoked at each iteration from $1$ to $k$. 

\begin{itemize}
\item \textsc{Mean Opinion}: First, we select the index $j^*$ such that
$$j^* = \underset{j \notin \Omega}{\text{argmax}} \Big|s_j' - \frac{1}{n}\sum_{i = 1}^n s_i'\Big|.$$

In words: among opinions that have not been changed yet, we choose the vertex whose opinion is closest to the current network's average opinion to be $j^*$. Second, we must change the opinion of $s_{j^*}'$ to $0$ or $1$. To do this, we optimize and set $s_{j^*}' = a^* = \underset{a \in \{0, 1\}}{\text{argmax}} \{f(s'): s_{j^*}' = a \}$. 

Note that the first step of the heuristic, which involves a larger decision space than the second, does not require any knowledge of the network structure or $L$, which can be the case in practice when edges are unknown; e.g., hiding a followers or friends list. Furthermore, if the adversary only has a rough idea of the nodes' opinions on average, this heuristic is intuitive: pick a ``centrist'' node with the most neutral opinion. The case when the adversary does not even have sufficient information about deciding the opinion $a^*$ leads us to define the next heuristic. 

\item \textsc{Mean Opinion (randomized)}: This heuristic is similar to the \textsc{Mean Opinion} heuristic, except the second step is replaced with randomly picking $s_{j^*}'$ to be equal to $0$ or $1$ with equal probability. This procedure can thus be entirely performed without knowledge of the underlying graph.

\item \textsc{Max Degree}: First, we select the index $j^*$ such that
\[
j^* = \underset{j \notin \Omega}{\text{argmax}} \sum_{i = 1}^n \ind{w_{i,j} > 0}.
\]

In words: we choose our vertex to be the one that is connected to the most other vertices in the network. Second, we set $a^* = \underset{a \in \{0, 1\}}{\text{argmax}} \{f(s'): s_{j^*}' = a \}$, as in \textsc{Mean Opinion}. 

Note that the first step of this heuristic, in contrast with \textsc{Mean Opinion}, does not require any knowledge of the innate opinion vector $s$ but rather exploits the network structure in a simple way. This heuristic is practical for adversaries that have access to the graph underlying the social network but may not have the means or data necessary to deduce what the opinions in the network are.

\item \textsc{Max Weighted Degree}: This heuristic is similar to \textsc{Max Degree}, except in the first step we choose $j^*$ whose sum of edge weights is the largest:
\[
j^* = \underset{j \notin \Omega}{\text{argmax}} \sum_{i = 1}^n w_{i,j}.
\]

\item \textsc{Random}: First, we select an opinion $j^* \notin \Omega$ to change uniformly at random. Second, we set $s_{j^*}'$ to either $0$ or $1$ with equal probability. 
This completely random algorithm offers a natural baseline to compare against. 
\end{itemize}

\subsection{Synthetic Experiments}

We evaluate the heuristics defined in Section~\ref{subsec:heuristics} over synthetic networks generated using three probabilistic models: the  Erd\H{o}s-R\'enyi model, the preferential attachment model~\cite{BA99}, and the stochastic block model~\cite{Abbe_survey}. The goal is to examine the empirical increase in disagreement and polarization that can be created by adversarial disruption in a variety of networks. Our empirical results show that several of the heuristics are able to increase the disagreement, polarization, and weighted sum of a network linearly in the budget $k$, which is the same order of growth as the upper bounds in Theorems~\ref{thm:P_bound} and~\ref{thm:D_bound}. For the weighted sum objective, we set $\lambda = 1$ in order to weigh polarization and disagreement equally.

Our empirical method is as follows. For each of the three models, we generate a network with $n = 1000$ vertices. Weights on the edges are randomly set to be values in $(0, 1]$ (and nonedges have zero weight). We experiment with $k$ in the range $0 \leq k \le n/2$, assuming it is unrealistic for an adversary to be able to change the opinions of more than half of the network. For each iteration until $n/2$, we plot the disagreement, polarization, and weighted sum when the adversary disrupts the network according to the six heuristics presented.

\paragraph{Erd\H{o}s-R\'enyi model.} 
In the Erd\H{o}s-R\'enyi model every pair of nodes is connected independently with some probability $p \in [0,1]$. This model serves as a natural null model for random graphs, with no underlying structure. In Figure~\ref{fig:erdos-renyi} we take $p = 0.2$; other values of $p$ show qualitatively similar behavior. 
We set the innate opinion vector $s$ to have i.i.d.\ values which are uniformly distributed in $[0, 1]$. 

Our results of the adversary determining $s'$ using the heuristics are shown in Figure \ref{fig:erdos-renyi}. We observe that all three objective functions are increasing roughly linearly in $k$, with Algorithm~\ref{alg:greedy} (greedy heuristic) performing the best. We also see that \textsc{Mean Opinion} and \textsc{Mean Opinion (Randomized)}, the two heuristics that exploit the innate opinion vector, are better than \textsc{Max Degree} and \textsc{Max Weighted Degree}, which exploit network structure. In fact, the latter two heuristics only appear slightly better than \textsc{Random} for all three objectives.

\begin{figure}[h!]
    \centering
    \includegraphics[scale=0.55]{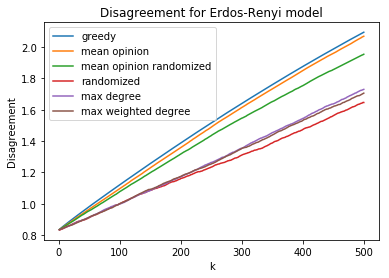}
    \includegraphics[scale=0.55]{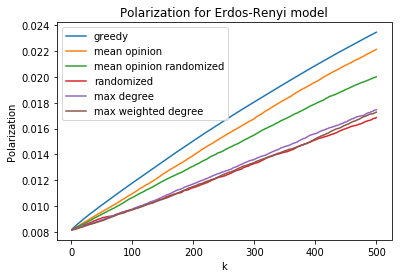}
    \includegraphics[scale=0.55]{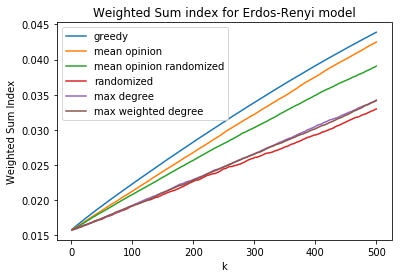}
    \caption{Performance of network disruption heuristics under the Erd\H{o}s-R\'enyi model with $p = 0.2$.}
    \label{fig:erdos-renyi}
\end{figure}

\paragraph{Preferential attachment model.} More realistic networks can be constructed with the preferential attachment process~\cite{BA99}. While the  Erd\H{o}s-R\'enyi random graph serves as a natural null model for a network with no structure, the preferential attachment process instead follows the natural concept that vertices that are more connected will receive more edges in the future. This is often true in social networks; for instance, new accounts on a social media platform are perhaps more likely to follow a popular account rather than a less known one. 
While the degree distribution of an Erd\H{o}s-R\'enyi random graph is binomial, preferential attachment graphs have a power-law degree distribution, which is often observed in real-world networks. 
We choose to generate a network using a preferential attachment process with parameter $m = 5$, meaning that at each time step, a new vertex is connected to $m$ existing nodes with a probability proportional to the degree of existing vertices. We again set the innate opinion vector $s$ to have i.i.d.\ values which are uniformly distributed in $[0, 1]$.

Our results are shown in Figure \ref{fig:pref}. Relatively, the greedy algorithm still has the best performance, followed by \textsc{Mean Opinion} and \textsc{Mean Opinion (Randomized)} for the $k$ defined in the synthetic experiments, and all heuristics seem to scale linearly in this range of $k$. We observe, however, that while \textsc{Max Degree} and \textsc{Max Weighted Degree} start out worse than pure randomization, they appear eventually to surpass \textsc{Random} and increase at a rate faster than other heuristics. Lastly, we observe that the scale of the objectives is significantly larger to start with than our results in Figure \ref{fig:erdos-renyi} using the Erd\H{o}s-R\'enyi model, perhaps due to the Erd\H{o}s-R\'enyi graph being much denser.

\begin{figure}[h]
    \centering
    \includegraphics[scale=0.55]{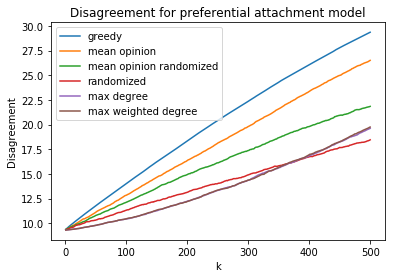}
    \includegraphics[scale=0.55]{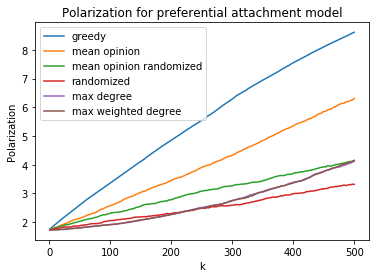}
    \includegraphics[scale=0.55]{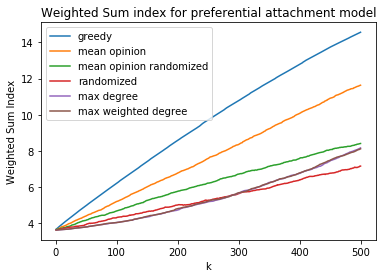}
    \caption{Performance of network disruption heuristics under the preferential attachment model with $m = 5$.}
    \label{fig:pref}
\end{figure}

\paragraph{Stochastic block model.} The stochastic block model \cite{Abbe_survey} is able to represent planted clusters, unlike the prior two models. These sorts of communities often arise in social networks, as seen in the Twitter data set \cite{musco2018minimizing} which we further discuss in Section~\ref{subsec:twitter_reddit}. We define two communities by partitioning $V$ into two sets $C_1$ and $C_2$, each of size $n/2$. Let the connectivity within both communities have parameters $p_{11} = p_{22} = 0.7$, that is, vertices within the communities share an edge with probability $0.7$ (independently across pairs), and let the connectivity between the two communities have parameter $p_{12} = 0.1$. Moreover, these communities often have different opinion distributions. Therefore, in our experiments we set the innate opinions $s_v$ for $v \in C_1$ to be independent draws from the $\Beta(5, 2)$ distribution, while the opinions of $s_v$ for $v \in C_2$ are i.i.d.\ $\Beta(2, 5)$. This means that opinions in $C_1$ are biased towards $1$, and opinions in $C_2$ are biased towards $0$. 
Experiments with different parameters show similar qualitative behavior as those discussed below. 

Our results are shown in Figure \ref{fig:sbm}. Similar to the two previous results, the greedy algorithm and \textsc{Mean Opinion} perform the best across the three objectives, increasing linearly in $k$. However, \textsc{Random} actually decreases the value of all three, and \textsc{Mean Opinion (Randomized)} decreases for polarization and the weighted sum. We conjecture that this is because choosing between $0$ and $1$ heavily depends on which community $j^*$ is in due to how the innate opinions are generated using two beta distributions rather than just a uniform distribution over $[0, 1]$.

\begin{center}
    \begin{figure}[h]
        \centering
        \includegraphics[scale=0.55]{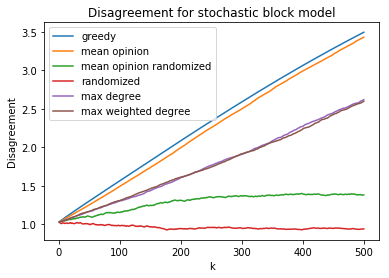}
        \includegraphics[scale=0.55]{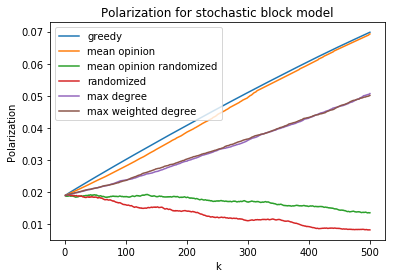}
        \includegraphics[scale=0.55]{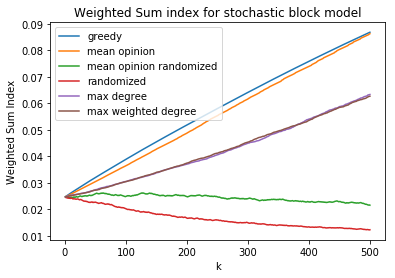}
        \caption{Performance of network disruption heuristics under the stochastic block model with $p_{11} = p_{22} = 0.7$, $p_{12} = 0.1$, and opinions distributed according to $\text{Beta}(5, 2)$ and $\text{Beta}(2, 5)$.}
        \label{fig:sbm}
    \end{figure}
\end{center}

\subsection{Analysis of Reddit and Twitter data sets} \label{subsec:twitter_reddit}

We now evaluate our proposed heuristics on two real data sets used in \cite{musco2018minimizing}. These data sets, one on Twitter and one on Reddit, contain the edgeset for the social networks as well as the list of opinions of the users over time. We pick the innate opinion vector to be the most recently recorded opinion vector, which is also how \cite{musco2018minimizing} chooses innate opinion vectors. They were originally collected by~\cite{de2014} by tracking interactions between users and using NLP techniques to map text to opinions.

\paragraph{Twitter.} This network has $n = 548$ vertices and $m = 3638$ edges, where the vertices represent the individuals tweeting over a certain time period about a debate on the Delhi legislative assembly elections of $2013$ (identified by a set of hashtags), and their opinions correspond to the sentiment of the tweets. Each edge is an undirected interaction between users.

Our results are shown in Figure \ref{fig:twitter}. Algorithm~\ref{alg:greedy} (greedy heuristic) and \textsc{Mean Opinion} still have the largest increases in all three objectives for this real data set. On the other hand, \textsc{Mean Opinion (Randomized)} and \textsc{Random} perform poorly, with \textsc{Max Degree} and \textsc{Max Weighted Degree} eventually outperforming the former two for all three objectives. This relative ordering of heuristic performance is similar to that of the stochastic block model discussed previously. In fact, when the Twitter edgeset is visualized, we can see that there are two main communities, and a third smaller and less dense community. Therefore, we can attribute a lot of the  performance results to the community structure. However, the distribution of innate opinions does not follow two beta distributions, but instead is approximately Gaussian with mean $0.602$ and standard deviation~$0.08$, which mitigates the decrease in performance that results from randomly setting $a^*$ amid beta-distributed opinions.

\begin{figure}[h]
    \centering
    \includegraphics[scale=0.55]{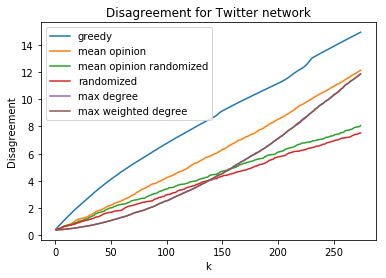}
    \includegraphics[scale=0.55]{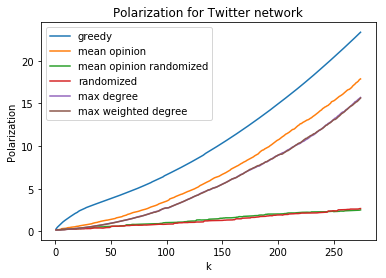}
    \includegraphics[scale=0.55]{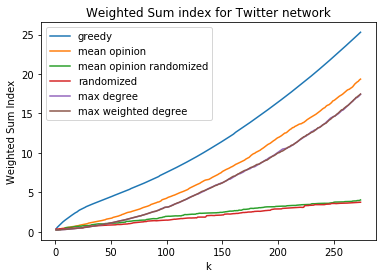}
    \caption{Performance of network disruption heuristics on a Twitter data set \cite{musco2018minimizing}.}
    \label{fig:twitter}
\end{figure}

In Table \ref{table:twitter}, we list the exact values for disagreement, polarization, and their weighted sum ($\lambda = 1$) of the Twitter network when the adversary uses the greedy algorithm, at the start of the algorithm ($k=0$) and when $k$ is equal to $20, 50, 100, 200$. This table suggests that, even if the adversary can only change the opinions on $20$ accounts (approximately $3.6\%$ of the nodes), the disagreement in the network increases by over $4$ times, while the polarization and weighted sum increase by over $7$ times.

\begin{table*}[h]
\centering
\begin{tabular}{|c|c|c c c c|}
     \hline 
     Objective & Original & $k = 20$ & $k = 50$ & $k = 100$ & $k = 200$  \\
     \hline
     Disagreement & $0.48$ & $2.12$ & $4.17$ & $6.81$ & $11.20$ \\
     Polarization & $0.29$ & $2.34$ & $3.89$ & $6.70$ & $15.05$ \\
     Weighted Sum & $0.37$ & $2.66$ & $4.48$ & $7.59$ & $16.54$ \\
     \hline
\end{tabular}
\caption{Values of objective functions for the Twitter data set under greedy heuristic (Algorithm~\ref{alg:greedy}) at $k = 0 \, (\text{original}), 20, 50, 100$, and $200$. }
\label{table:twitter}
\end{table*}

\paragraph{Reddit.} This network has $n = 556$ vertices and $m = 8969$ edges, where the vertices represent individuals who have posted in a politics subreddit, and their opinions correspond to the sentiment in this subreddit over a certain time period. There is an edge between users if they both post in at least two other same subreddits. We also discard three vertices from this graph that are not connected to any other vertices, as keeping these vertices implies that heuristics can simply change these opinions to yield large increases in polarization without any consequences for the opinion dynamics.

Our results are shown in Figure \ref{fig:reddit}. Again, the greedy algorithm performs best, with a large increase especially for small $k$. While the graphs for polarization and for the weighted sum have very noticeable jumps, for all three objectives \textsc{Mean Opinion}, \textsc{Mean Opinion (Randomized)}, and \textsc{Random} perform similarly. We conjecture that random is not the worst in this case for two reasons: firstly, the Reddit data set's opinions roughly follow a normal distribution with mean $0.498$ and standard deviation $0.04$, meaning that the values are more tightly concentrated around a very neutral opinion than the Twitter data set. Moreover, the distribution of degrees of the vertices is more uniform than that of the Twitter data set, which appears to follow a power law instead, suggesting that arbitrarily choosing a vertex and then randomly setting its opinion can still result in good performance. 

\begin{figure}[h]
    \begin{center}
        \includegraphics[scale=0.55]{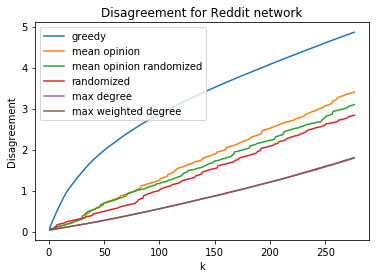}
        \includegraphics[scale=0.55]{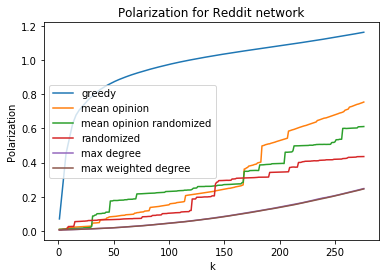}
        \includegraphics[scale=0.55]{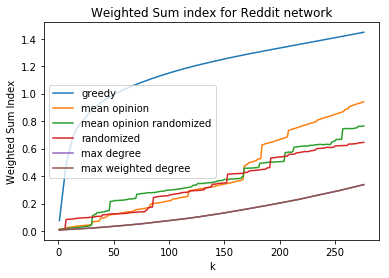}
    \end{center}
    \caption{Performance of network disruption heuristics on a Reddit data set \cite{musco2018minimizing}.}
    \label{fig:reddit}
\end{figure}

In Table \ref{table:reddit}, we list the exact values for disagreement, polarization, and their weighted sum ($\lambda = 1$)  of the Reddit network when the adversary uses the greedy algorithm at the start of the algorithm ($k=0$) and when $k$ is equal to $20, 50, 100, 200$. This table suggests that, even if the adversary can only change the opinions on $20$ accounts (approximately $3.6\%$ of the nodes), all objectives are able to increase roughly tenfold.

\begin{table*}[h]
\centering
\begin{tabular}{|c|c|c c c c|}
     \hline 
     Objective & Original & $k = 20$ & $k = 50$ & $k = 100$ & $k = 200$  \\
     \hline
     Disagreement & $0.09$ & $1.14$ & $2.00$ & $2.88$ & $4.09$ \\
     Polarization & $0.07$ & $0.72$ & $0.88$ & $0.98$ & $1.09$ \\
     Weighted Sum ($\lambda = 1$) & $0.08$ & $0.79$ & $0.99$ & $1.15$ & $1.33$ \\
     \hline
\end{tabular}
\caption{Values of objective functions for the Reddit data set under greedy heuristic (Algorithm~\ref{alg:greedy}) at $k = 0 \, (\text{original}), 20, 50, 100,$ and $200$.}
\label{table:reddit}
\end{table*}

\section{Conclusion and discussion} \label{sec:discussion} 

Our primary conceptual contribution in this paper is the introduction of a new, adversarial model of network disruption, detailed in the Introduction and illustrated in Figure~\ref{fig:schematic}. 
We investigated this model both theoretically (Sections~\ref{sec:convexity} and~\ref{sec:bounds}) and empirically (Section~\ref{sec:algorithms}). 
The key take-away from our results is that an adversary can significantly increase disagreement / polarization, even using simple, unsophisticated methods. 
This motivates further research into addressing these challenges. 

We conclude by highlighting some specific and some broad questions that our paper leaves open.

\begin{itemize}
    \item \textbf{Hardness of optimal network disruption?} As mentioned in the introduction, we conjecture that solving the optimization problem of the adversary is computationally hard when $k$ is large. Recent work of Gionis, Terzi, and Tsaparas~\cite{gionis2013opinion} on a related opinion maximization problem uses a reduction to vertex cover to show hardness; see also~\cite{abebe2018opinion} where this proof is adapted to another setting. Adapting this proof to our setting is challenging due to the different nature of our objective function, coupled with the opinion dynamics whose effect is difficult to isolate. 
    \item \textbf{Performance guarantees for the adversary?} In Section~\ref{sec:algorithms} we investigated empirically the performance of several natural algorithms for the adversary, on several different random graphs, as well as on Reddit and Twitter data sets. While performances varied, depending on the algorithm and the underlying social network, one thing that they all had in common was a linear growth in the objective function, as a function of the budget $k$. Is it possible to prove such a performance guarantee?
    \item \textbf{Other opinion dynamics?} We focused here on the Friedkin-Johnsen model of opinion dynamics, but everything we discussed can be studied under other models. How robust are the results to such changes?
    
    \item \textbf{Defense strategies?} Our empirical results in Section~\ref{sec:algorithms} show that the adversary does not have to be sophisticated in order to significantly disrupt the network. This highlights the need to think carefully about defense strategies that can counteract network disruption. For instance, is it possible to tackle network disruption by modifying the network itself (e.g., by carefully suggesting new edges to add)?
\end{itemize}


\section*{Acknowledgements}

We thank Cameron Musco, Christopher Musco, and Charalampos Tsourakakis for sharing the Reddit data set used in~\cite{musco2018minimizing} and for clarifying details of their data analysis. 


\bibliographystyle{plain}
\bibliography{bib}




\end{document}